\documentclass[11pt]{article}
\usepackage{cmap}
\usepackage[T1]{fontenc}
\usepackage{lmodern}
\usepackage{microtype}
\usepackage[margin=1in]{geometry}
\usepackage{amsmath,amssymb,amsthm}
\usepackage{graphicx}
\usepackage{booktabs}
\usepackage[font=small]{caption}
\usepackage{placeins}
\usepackage[round]{natbib}
\usepackage{hyperref}
\hypersetup{
  hidelinks,
  pdftitle={When can a posterior predictive check identify the learning rate?},
  pdfauthor={Nam Anh Le},
  pdfsubject={Generalised Bayesian inference; posterior predictive checks; learning-rate selection},
  pdfkeywords={generalised Bayesian inference, posterior predictive checks, learning rate, Gaussian linear models}
}
\graphicspath{{./}}

\newtheorem{theorem}{Theorem}
\newtheorem{corollary}{Corollary}

\theoremstyle{remark}
\newtheorem{remark}{Remark}

\newcommand{\RSS}{\mathrm{RSS}}
\newcommand{\E}{\mathbb{E}}
\newcommand{\R}{\mathbb{R}}
\newcommand{\Prob}{\mathrm{P}}
\newcommand{\IG}{\mathrm{IG}}
\newcommand{\bbeta}{\beta}
\newcommand{\bhat}{\hat{\bbeta}}
\newcommand{\chisq}{\chi^2}

\DeclareMathOperator{\Var}{Var}

\title{\textbf{When can a posterior predictive check identify the learning rate?\\[0.25em]
Exact degeneracy in Gaussian models and implications\\
for Generalised Bayesian Inference}}
\author{Nam Anh Le\\
\small National Economics University, Hanoi, Vietnam\\
\small \texttt{me@namanhle.com}}
\date{June 4, 2026}

\begin{document}
\maketitle

\begin{abstract}
Generalised Bayesian inference tempers the likelihood by a learning rate $\eta$ to mitigate model
misspecification, and the choice of $\eta$ is consequential. \citet{zafar2024} proposed selecting $\eta$ by a
posterior predictive check (PPC): one chooses the smallest $\eta$ at which a log-likelihood PPC $p$-value is not
rejected. An exact, finite-sample analysis of this selector on the Gaussian linear model is given. With known
variance and a flat prior, the PPC $p$-value equals $\Prob(\chisq_n > \RSS/\sigma_0^2)$ for \emph{every} $\eta$,
so the selector is $\eta$-invariant; under variance misspecification it is two-sided non-identifying
($p\to 0$ when the data are over-dispersed, $p\to 1$ when under-dispersed). With unknown variance and the
reference prior, the $p$-value admits the closed form
$p_{n,d}(\eta)=\Prob[\,W+n\log\{(\eta n-d-2)/V\}>(\eta n-d-2)/\eta\,]$ with $W\sim\chisq_n$ independent of
$V\sim\chisq_{\eta n-d}$; remarkably this depends only on $(n,d,\eta)$ and \emph{not on the realised data or the
data-generating process}. Consequently the selector's output is fixed before any data are seen, typically
collapsing to the smallest grid value, which over-tempers and inflates predictive intervals relative to held-out
selection. The phenomenon is a pivotality property specific to the Gaussian scale--location family and the
reference prior; it disappears under informative priors. These results delineate the
selector's scope, identify a canonical class on which it cannot identify the learning rate, and motivate a cheap, data-free
pre-screening diagnostic.
\end{abstract}

\section{Introduction}
Standard Bayesian updating can behave poorly under model misspecification. Generalised (or Gibbs) Bayesian
inference replaces the likelihood $p(y\mid\theta)$ by a tempered version $p(y\mid\theta)^{\eta}$, where the
learning rate $\eta>0$ down-weights the likelihood \citep{bissiri2016,grunwald2017}. The value of $\eta$
materially affects calibration, and many selection rules exist \citep{holmes2017,lyddon2019,syring2019,wu2023,
lee2025}. \citet{zafar2024} recently proposed a strikingly simple device: choose $\eta$ using a
\emph{posterior predictive check}. They report promising results on a complex, misspecified model for diachronic
sense change, while emphasizing that the approach is exploratory and does not yet come with universal
guarantees. They do not examine simple or conjugate models.

This note asks a narrow but basic question: \emph{on the canonical Gaussian linear model, can this PPC $p$-value
distinguish learning rates at all?} In an exact finite-sample sense, it cannot. The contribution is
threefold: (i) closed-form laws for the selector's $p$-value (Theorems
\ref{thm:known}, \ref{thm:unknown}); (ii) an \emph{identifiability} reading---under known variance the $p$-value
is constant in $\eta$, and under unknown variance with the reference prior it is constant in the \emph{data}, so
the selected $\eta$ is determined before any data are observed; and (iii) a data-free pre-screening diagnostic.
This is a boundary result rather than a general negative statement: it leaves the EDiSC findings of
\citet{zafar2024} untouched while isolating a clean class on which the selector is provably non-identifying, and
explains why.

The mechanism---pivotality of a realised, parameter-dependent discrepancy---is classical in spirit
\citep{meng1994,gelman1996}; posterior predictive $p$-values are known to be non-uniform and to concentrate
(asymptotically) away from the extremes \citep{robins2000}, and their distributions are exactly those dominated
in convex order by the uniform \citep{rdl2015}. The finite-sample, data-free statements here are consistent with and
sharpen these facts in a specific modern context, and connect them to learning-rate selection, where they have
not previously been noted. A parallel ``temperature invisibility'' phenomenon---tempering not affecting posterior
\emph{predictions} in large samples---was recently established by \citet{mclatchie2024}.

\section{The PPC learning-rate selector}
Let $y\in\R^n$ be modelled as $y=X\bbeta+\varepsilon$, $\varepsilon\sim N(0,\sigma^2 I_n)$, with
$X\in\R^{n\times d}$ of rank $d$. For learning rate $\eta>0$ the power (Gibbs) posterior is
$\pi_\eta(\theta\mid y)\propto \pi(\theta)\,p(y\mid\theta)^{\eta}$. Following \citet{zafar2024} (and their
reference implementation), let $\hat\theta_\eta=\E_{\pi_\eta}[\theta\mid y]$ be the posterior mean and define the
log-likelihood discrepancy $T(z,\theta)=\log p(z\mid\theta)$. The PPC $p$-value is
\begin{equation}\label{eq:ppc}
p(\eta)=\Prob\big(T(y^{\mathrm{rep}},\theta) < T(y,\hat\theta_\eta)\big),
\qquad \theta\sim\pi_\eta(\cdot\mid y),\ \ y^{\mathrm{rep}}\mid\theta\sim p(\cdot\mid\theta),
\end{equation}
where the replicate is drawn from the \emph{untempered} model; a small $p(\eta)$ indicates misfit. Given a finite
grid $\mathcal{G}$ and level $\alpha$ (Zafar and Nicholls use $\alpha=0.10$), the selector returns
\begin{equation}\label{eq:selector}
\eta^\star=\min\{\eta\in\mathcal G: p(\eta)\ge\alpha\}.
\end{equation}
In the unknown-variance calculations below, the effective grid is
$\mathcal G_+=\{\eta\in\mathcal G:\eta n>d+2\}$, because the posterior mean of $\sigma^2$ is finite only on this
set.
Write $\bhat=(X^\top X)^{-1}X^\top y$ and $\RSS=\lVert y-X\bhat\rVert^2$, and let $\Prob(\chisq_k>\cdot)$ denote
the upper tail of a chi-square with $k$ degrees of freedom.

\section{Exact degeneracy under known variance}
\begin{theorem}\label{thm:known}
Suppose $\sigma^2=\sigma_0^2$ is known and $\pi(\bbeta)\propto 1$. Then for every $\eta>0$,
\[
p(\eta)=\Prob\!\big(\chisq_n > \RSS/\sigma_0^2\big).
\]
In particular $p(\cdot)$ is constant in $\eta$, and the selector \eqref{eq:selector} is degenerate: it accepts
all of $\mathcal G$ (returning $\min\mathcal G$) when $\RSS/\sigma_0^2\le \chisq_{n,1-\alpha}$, and rejects all of
$\mathcal G$ otherwise.
\end{theorem}
\begin{proof}
The power posterior is $\bbeta\mid y\sim N(\bhat,(\sigma_0^2/\eta)(X^\top X)^{-1})$, so its mean
$\hat\theta_\eta=\bhat$ does not depend on $\eta$ and
$T(y,\hat\theta_\eta)=-\tfrac n2\log(2\pi\sigma_0^2)-\RSS/(2\sigma_0^2)$. Conditional on a draw $\bbeta$, the
replicate satisfies $y^{\mathrm{rep}}-X\bbeta\sim N(0,\sigma_0^2 I_n)$, hence
$\lVert y^{\mathrm{rep}}-X\bbeta\rVert^2/\sigma_0^2\sim\chisq_n$ independently of $\bbeta$, and
$T(y^{\mathrm{rep}},\bbeta)\stackrel{d}{=}-\tfrac n2\log(2\pi\sigma_0^2)-\tfrac12 W$ with $W\sim\chisq_n$. Thus
$p(\eta)=\Prob(\tfrac12 W>\RSS/(2\sigma_0^2))=\Prob(\chisq_n>\RSS/\sigma_0^2)$.
\end{proof}

The case $d=1$, $X=\mathbf 1_n$ is the Gaussian location model; Theorem \ref{thm:known} covers the entire linear
class (any design, any $d$). The selector is non-identifying not only when the model is correct but in both
directions of variance misspecification.

\begin{corollary}[two-sided non-identifiability]\label{cor:twosided}
Assume the mean is correctly specified and let $r=\Var_{\mathrm{DGP}}(Y)/\sigma_0^2$. Then $\RSS/\sigma_0^2$ has
mean $(n-d)r$ and concentrates, so as $n\to\infty$ with $d$ fixed, $p(\eta)\to 0$ if $r>1$ and $p(\eta)\to1$ if
$r<1$, for every $\eta$. The selector returns no accepted grid point when the working variance is too small and the grid floor
when it is too large; in neither case does it use $\eta$.
\end{corollary}

\section{Unknown variance: an exact, data-free law}
\begin{theorem}\label{thm:unknown}
Suppose $\bbeta,\sigma^2$ are unknown with reference prior $\pi(\bbeta,\sigma^2)\propto\sigma^{-2}$, and
$\eta n>d+2$. Then $\sigma^2\mid y\sim\IG\!\big(\tfrac{\eta n-d}{2},\tfrac{\eta\,\RSS}{2}\big)$,
$\E[\sigma^2\mid y]=\eta\,\RSS/(\eta n-d-2)$, and
\begin{equation}\label{eq:law}
p(\eta)=p_{n,d}(\eta):=\Prob\!\Big[\,W+n\log\frac{\eta n-d-2}{V}>\frac{\eta n-d-2}{\eta}\,\Big],
\qquad W\sim\chisq_n \perp V\sim\chisq_{\eta n-d}.
\end{equation}
The right-hand side depends only on $(n,d,\eta)$: it is independent of the realised data and of the
data-generating process.
\end{theorem}
\begin{proof}[Proof sketch]
Decompose $\lVert y-X\bbeta\rVert^2=\RSS+(\bbeta-\bhat)^\top X^\top X(\bbeta-\bhat)$. Integrating $\bbeta$ from
$\pi_\eta(\bbeta,\sigma^2\mid y)\propto(\sigma^2)^{-\eta n/2-1}\exp\{-\tfrac{\eta}{2\sigma^2}\lVert
y-X\bbeta\rVert^2\}$ yields the stated inverse-gamma marginal and $\bbeta\mid\sigma^2\sim N(\bhat,
(\sigma^2/\eta)(X^\top X)^{-1})$, so $\hat\theta_\eta=(\bhat,\hat\sigma^2_\eta)$ with
$\hat\sigma^2_\eta=\eta\,\RSS/(\eta n-d-2)$ and $\RSS/\hat\sigma^2_\eta=(\eta n-d-2)/\eta$. Hence
$T(y,\hat\theta_\eta)=-\tfrac n2\log(2\pi\hat\sigma^2_\eta)-\tfrac12(\eta n-d-2)/\eta$. For a replicate,
$\lVert y^{\mathrm{rep}}-X\bbeta\rVert^2/\sigma^2=:W\sim\chisq_n$ given the draw, so
$T(y^{\mathrm{rep}},\theta)=-\tfrac n2\log(2\pi\sigma^2)-\tfrac12 W$. Since
$\eta\,\RSS/\sigma^2=:V\sim\chisq_{\eta n-d}$, this gives $\sigma^2=\eta\,\RSS/V$ and
$\log(\sigma^2/\hat\sigma^2_\eta)=\log\{(\eta n-d-2)/V\}$, in which $\RSS$ cancels. Substituting into
$\Prob(T(y^{\mathrm{rep}},\theta)<T(y,\hat\theta_\eta))$ gives \eqref{eq:law}; $W\perp V$ because $W$ is a fresh
replicate quantity whose law does not depend on the posterior draw.
\end{proof}

\begin{corollary}[data-free selection]\label{cor:datafree}
Under the hypotheses of Theorem \ref{thm:unknown}, the selector \eqref{eq:selector} is a deterministic function
of $(n,d,\mathcal G_+,\alpha)$ alone. Numerically (Section \ref{sec:num}), once the grid is away from the propriety
edge $\eta=(d+2)/n$, the curve $\eta\mapsto p_{n,d}(\eta)$ typically lies well above $\alpha=0.10$ across the
admissible grid; hence $\eta^\star=\min\mathcal G_+$.
\end{corollary}

\begin{remark}[robustness to the variance plug-in]\label{rem:plugin}
The data-independence does not rely on using the posterior \emph{mean} of $\sigma^2$. Any summary of the
inverse-gamma posterior used in $T(y,\cdot)$ has the form $\hat\sigma^2_\eta=\kappa(\eta,n,d)\,\RSS$ (true for
the mean, mode, median, and any quantile, since the scale is $\propto\RSS$); the cancellation
$\log(\sigma^2/\hat\sigma^2_\eta)=\log\{\eta\RSS/(\kappa\RSS V)\}$ removes $\RSS$ regardless, changing only an
additive constant. The phenomenon is structural, not an artefact of the mean.
\end{remark}

\paragraph{Prior sensitivity.} The exactness is specific to the reference prior. With a proper ridge prior
$\bbeta\sim N(0,\tau^2 I)$ (known variance), the posterior mean is the $\eta$-dependent ridge estimate
$\bhat_\eta=(X^\top X+\tfrac{\sigma_0^2}{\eta\tau^2}I)^{-1}X^\top y$, so
$p(\eta)=\Prob(\chisq_n>\lVert y-X\bhat_\eta\rVert^2/\sigma_0^2)$ does depend on $\eta$. The dependence is
governed by the prior strength $\sigma_0^2/(\eta\tau^2)$, vanishes as $\tau^2\to\infty$, and is numerically
negligible for diffuse priors, so the selector still collapses to the grid floor over standard grids; only
strong, informative priors restore appreciable data-dependence, and that dependence reflects prior--data tension
rather than misspecification detection.

\section{Numerical implications for calibration}\label{sec:num}
The numerical experiments illustrate the two degeneracies and their calibration consequence. Figure \ref{fig:1}
confirms the $\eta$-invariance and the two-sided transition. Figure \ref{fig:2} shows that Monte Carlo $p$-values
from different data-generating processes land exactly on the single data-free curve $p_{n,d}(\eta)$. Figure
\ref{fig:3} maps the data-free selected $\eta$ over $(n,d)$. Figure \ref{fig:4} quantifies the calibration cost:
because the PPC selects the grid floor regardless of the data, the resulting $90\%$ predictive intervals over-cover
($0.97$--$0.99$ against nominal $0.90$) and are $1.5$--$1.9\times$ wider than those from held-out (ELPPD) selection
or from $\eta=1$, across well-specified, heavy-tailed, and contaminated data.

\begin{figure}[!htbp]\centering
\includegraphics[width=\textwidth]{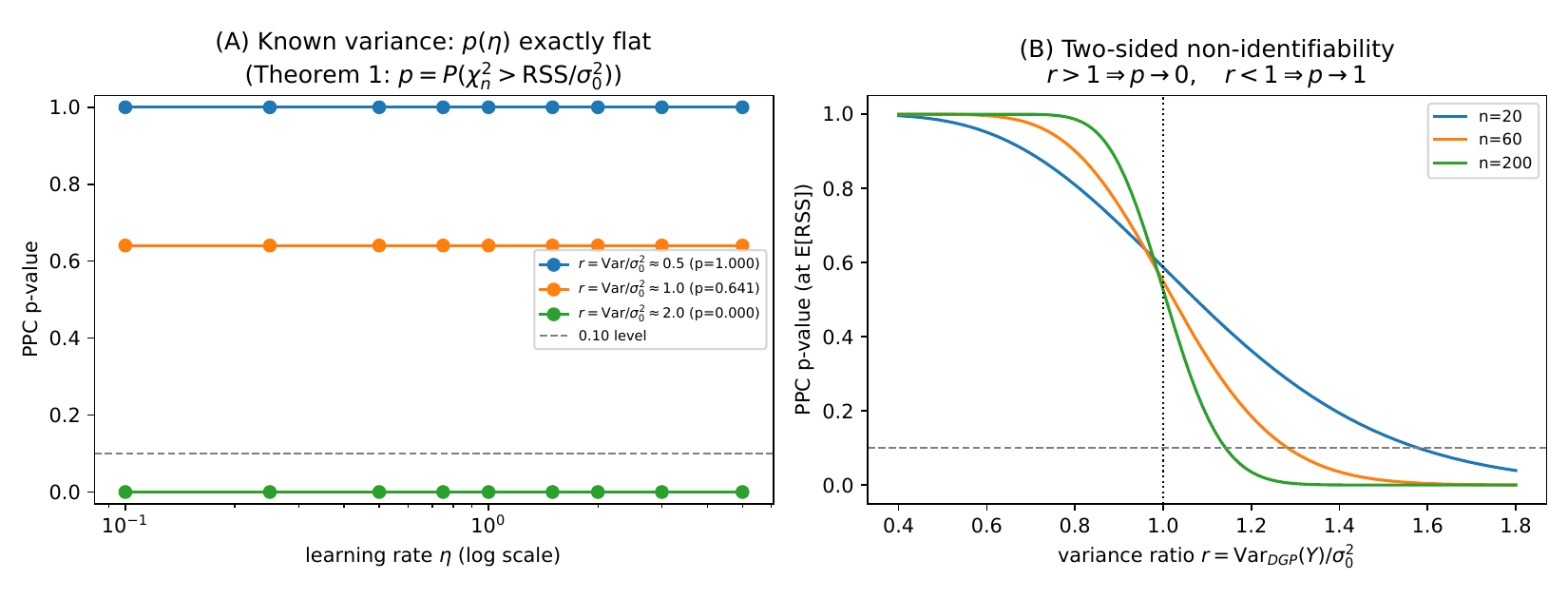}
\caption{Known variance. (A) The PPC $p$-value is exactly flat in $\eta$ (Theorem \ref{thm:known}); each line is a
different variance ratio $r$. (B) Two-sided non-identifiability: $p\to1$ for $r<1$ and $p\to0$ for $r>1$
(Corollary \ref{cor:twosided}).}\label{fig:1}
\end{figure}

\begin{figure}[!htbp]\centering
\includegraphics[width=0.72\textwidth]{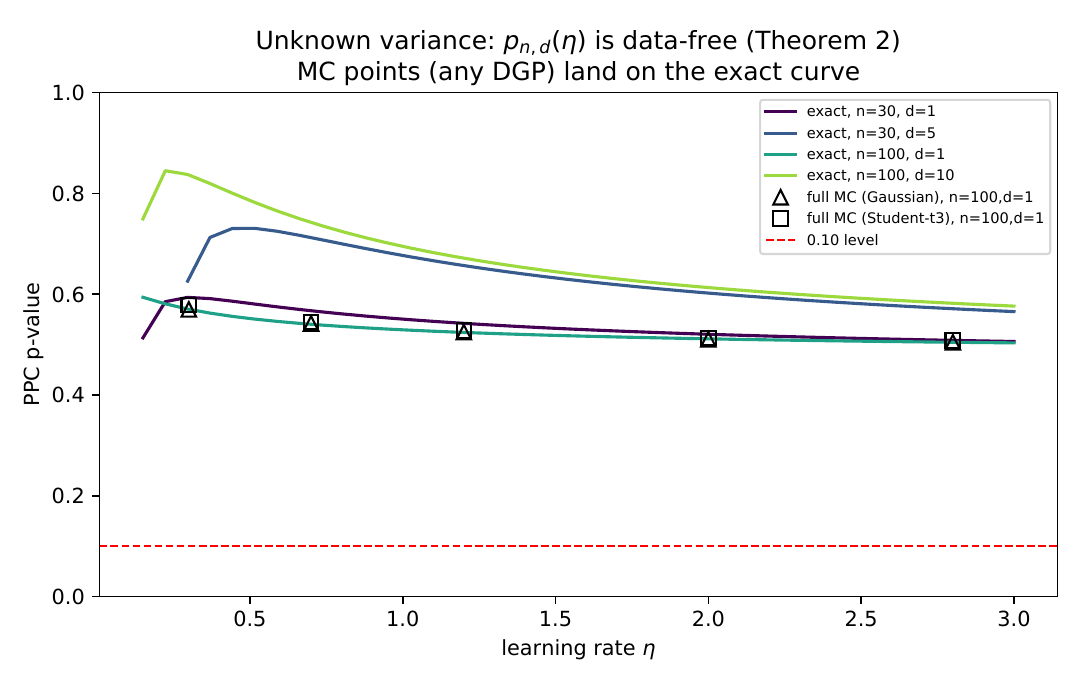}
\caption{Unknown variance. Solid: the exact data-free curve $p_{n,d}(\eta)$ for several $(n,d)$. Markers:
Monte Carlo $p$-values from a Gaussian and a Student-$t_3$ DGP ($n=100,d=1$), which coincide with the curve,
illustrating data and DGP independence (Theorem \ref{thm:unknown}).}\label{fig:2}
\end{figure}

\begin{figure}[!htbp]\centering
\includegraphics[width=0.92\textwidth]{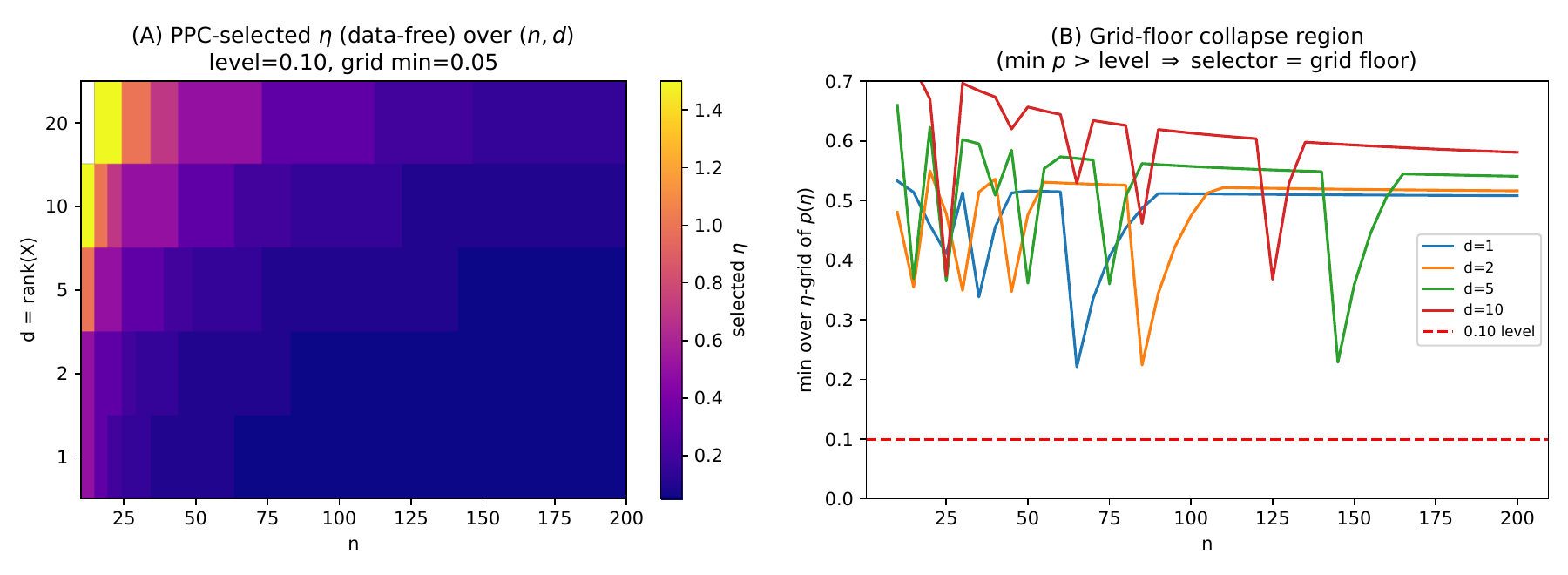}
\caption{(A) The PPC-selected $\eta$ is data-free and equals the grid floor across most of the $(n,d)$ plane.
(B) $\min_{\eta\in\mathcal G_+}p_{n,d}(\eta)$ stays above the level $\alpha=0.10$, forcing grid-floor
collapse.}\label{fig:3}
\end{figure}

\begin{figure}[!htbp]\centering
\includegraphics[width=0.92\textwidth]{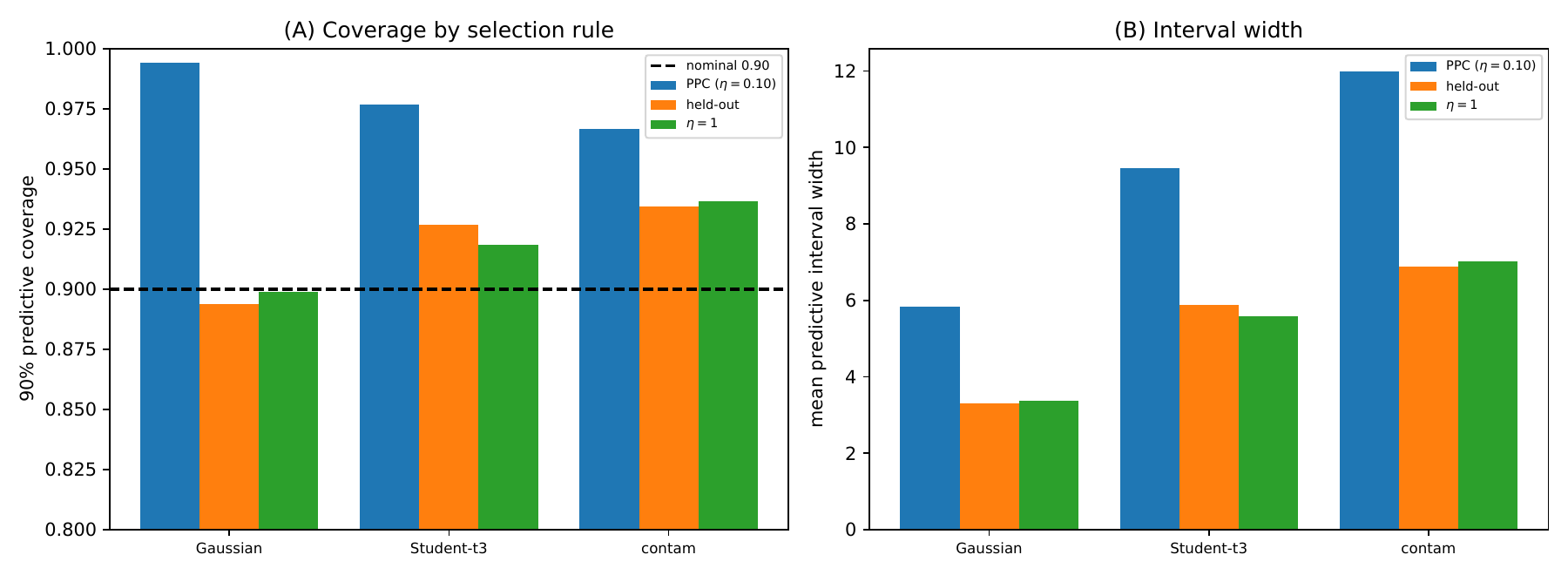}
\caption{Calibration cost ($n=80,d=3$, $90\%$ intervals). The data-free PPC choice $\eta=0.1$ over-covers and
inflates width in every DGP, whereas held-out selection and $\eta=1$ track the nominal level.}\label{fig:4}
\end{figure}

\FloatBarrier

\section{A practical diagnostic, and discussion}
\paragraph{Pre-screening diagnostic.} Theorem \ref{thm:unknown} suggests a cheap safeguard. For an
(approximately) Gaussian-linear model with the reference prior, compute---\emph{before fitting}---the data-free
quantities $\underline p=\min_{\eta\in\mathcal G_+}p_{n,d}(\eta)$ and the range
$\rho=\max_{\eta\in\mathcal G_+}p_{n,d}(\eta)-\min_{\eta\in\mathcal G_+}p_{n,d}(\eta)$. If $\underline p>\alpha$
(equivalently, the curve never enters the rejection region), the log-likelihood PPC selector has no data-driven
role in this regime; a
data-driven alternative such as held-out/ELPPD selection \citep{lee2025} or generalized posterior calibration
\citep{wu2023} is preferable. The diagnostic costs a one-dimensional quadrature per grid point and requires no model
fit.

\paragraph{Relation to classical results.} The effect is an instance of pivotality: under the Gaussian
scale--location group and the reference prior, both the observed and replicate discrepancies are functions of
pivotal quantities, so their comparison probability is free of the unknown parameters and (in the unknown-variance
case) of the data. This is in the spirit of \citet{meng1994} and \citet{gelman1996}, consistent with the
asymptotic concentration of \citet{robins2000}, and a clean instance of the convex-order characterization of
\citet{rdl2015} (here the sampling law of $p(\eta)$ is degenerate---a point mass---the extremal convex-order
case). What is new is the exact finite-sample form, the data-independence, and the consequence for
learning-rate \emph{selection}.

\paragraph{Scope and limitations.} The results target a canonical Gaussian scale--location regime rather than
attempting to classify all PPC learning-rate selectors. The exact invariance uses the flat/reference prior and the
Gaussian likelihood's pivotal residual and scale structure; informative priors or non-Gaussian likelihoods introduce
additional data-dependence. Such data-dependence should be assessed directly, because it does not automatically make
the log-likelihood PPC a reliable learning-rate selector; it only moves the problem outside the exact pivotal regime
analysed here. The practical implication is a screening rule: when a fitted model is close to this regime, the
selector can be recognized in advance as data-free or nearly so, and held-out/ELPPD or calibration-based alternatives
are better justified.

\paragraph{Acknowledgements.} Computations are fully reproducible from the accompanying package (seed
\texttt{20260603}).

\end{document}